\documentclass[runningheads]{llncs}

\usepackage{amsthm}
\usepackage{amsfonts}
\usepackage{mathtools}
\usepackage{mathabx}

\usepackage{comment}
\usepackage{multirow}
\usepackage{array}
\usepackage{arydshln}
\usepackage{float}
\usepackage{bm}
\usepackage{multicol}
\usepackage{appendix}

\usepackage{ifthen}
\newboolean{isFullVersion}
\setboolean{isFullVersion}{true}

\newcommand{\varA}[1]{{\operatorname{\mathit{#1}}}}

\begin{document} 

\title{A Combinatorial Characterization of Self-Stabilizing Population Protocols}

\author{Shaan Mathur \and Rafail Ostrovsky}

\institute{University of California, Los Angeles \email{\{shaan, rafail\}@cs.ucla.edu}}



\maketitle

\begin{abstract}
    We characterize self-stabilizing functions in population protocols for complete interaction graphs. In particular, we investigate self-stabilization in systems of $n$ finite state agents in which a malicious scheduler selects an arbitrary sequence of pairwise interactions under a global fairness condition. We show a necessary and sufficient condition for self-stabilization. Specifically we show that functions without certain set-theoretic conditions are impossible to compute in a self-stabilizing manner. Our main contribution is in the converse, where we construct a self-stabilizing protocol for all other functions that meet this characterization. Our positive construction uses Dickson's Lemma to develop the notion of the root set, a concept that turns out to fundamentally characterize self-stabilization in this model. We believe it may lend to characterizing self-stabilization in more general models as well.
    
\end{abstract}



\newtheorem*{theorem*}{Theorem}
\newtheorem*{lemma*}{Lemma}
\newtheorem*{corollary*}{Corollary}





\section{Introduction}
The population protocol computational model assumes a system of $n$ identical finite state transducers (which we call {\em agents}) in which pairwise interactions between agents induce their respective state transitions. Each agent is provided a starting input and starting state, and an adversarial scheduler decides at each time step which two agents are to interact. In order to make the behavior of the scheduler precise, the scheduler is allowed to act arbitrarily so long as the \textit{global fairness condition} (defined in the paper defining the model by Angluin, Aspnes, Diamadi, Fischer, and Peralta \cite{Model}) is satisfied: if a configuration of agent states appears infinitely often, then any configuration that can follow (say, after an interaction) must also appear infinitely often. We stress that agents individually do not have unique identifiers and a bound on $n$ in not known. We call all $n$ agents jointly a {\em population}. When an agent interacts with another agent, both agents change states as a function of each agent's input and state tuple. Each agent outputs some symbol at every time step as a function of their current state, and once every agent agrees on some common output for all subsequent time steps, we say the protocol has \textit{converged} to that symbol. We say a protocol computes (or decides) some function $f$ if distributing the input symbols of input $x$ and running the protocol causes the population (i.e. all agents) to converge to $f(x)$. In the general model an accompanying \textit{interaction graph} restricting which agents can ever interact may be provided as a constraint on the scheduler. In this paper, we will be considering the original, basic model introduced in \cite{Model}, where we deal with complete interaction graphs and inputs that do not change with time. For the population protocol model on complete graphs, the characterization of computable predicates (not necessarily with self-stabilization) has been studied by Angluin, Aspnes, Eisenstat, and Ruppert, who proved it to be equivalent to the set of semilinear predicates \cite{Semilinear}. 

It is desirable to have population protocols that can handle transient faults: specifically, we would like it to be the case that no matter what multiset of states the agents are initialized with, the protocol will eventually converge to the correct output. For this reason, such a protocol is called \textit{self-stabilizing}, being able to converge after experiencing any adversarial fault that erroneously changes an agent state. Since fault-tolerance is a desirable property of any distributed system, we aim to determine exactly what computable functions in this model admit self-stabilizing solutions, and which do not. We prove the following main theorem for population protocols on complete interaction graphs.

\begin{theorem*}[Main]\label{Self-Stabilizing Population Protocol Theorem}
    Let $f: \mathcal{X} \rightarrow Y$ be a function where $\mathcal{X}$ is any set of finite multisets on a finite alphabet. Then
        $f \text{ has a self-stabilizing protocol} \iff (\text{For any } A,B \in \mathcal{X}, A \subseteq B \implies f(A) = f(B)).$
\end{theorem*}

\medskip

We remark on the definition of $\mathcal{X}$: Firstly, it is a set of multisets; thus $A \subseteq B$ refers to multiset inclusion, not set inclusion (e.g. $\{a, a, b\} \subseteq \{a, a, b, b \}$ but $\{a, a, b\} \not \subseteq \{a, b\}).$ Moreover, the domain $\mathcal{X}$ can be \textit{any} set of multisets; that is, the domain does not necessarily contain all possible nonempty finite multisets on a finite alphabet, but could be a subset of it. If $\mathcal{X}$ is all possible nonempty finite multisets on a finite alphabet, then the theorem states $f$ is a constant function: the output of $f$ on the singleton multisets is the output of $f$ on the union of all singletons; since $f$ agrees on all singletons, it agrees on all larger multisets. Thus self-stabilizing decision problems, where all possible inputs are included, will be a constant function. In contrast, self-stabilizing promise problems, where only a subset of all possible inputs are included, may be non-constant. 

\noindent {\bf Our Techniques:} The technique that we use to show when self-stabilization is not possible follows from the work of Angluin, Aspnes, Fischer, and Jiang in \cite{Angluin}. Informally, self-stabilizing functions must not allow subpopulations to ``re-converge'' to a different answer; so if $A \subseteq B$ but $f(A) \neq f(B)$, running a protocol with input $B$ from any configuration could lead to the subpopulation with input $A$ converging on erroneous output $f(A)$. The converse, that functions where subsets lead to the same output are self-stabilizing, is an unstudied problem that involves a technically intricate construction. Tools from partial order theory (Dickson's Lemma) are used to observe that the domain $\mathcal{X}$ will have a finite set of minimal elements under the $\subseteq$ partial order; these minimal elements completely determine the output of $f$, and so our protocol computes $f$ by identifying which of these minimal elements is present in the population. 
A function is {\em self-stabilizing} if it admits a self-stabilizing protocol under the basic model. A nontrivial corollary of our main theorem is that for a fixed (possibly infinite) domain $\mathcal{X}$ and a finite output alphabet $Y$, there are only finitely many self-stabilizing functions $f: \mathcal{X} \rightarrow Y$.

\begin{corollary*}[Self-Stabilizing Functions are Rare]
    Fix some set, $\mathcal{X}$, of finite multisets on a finite alphabet. Fix a finite output alphabet $Y$. There are only finitely many self-stabilizing functions of the form $f: \mathcal{X} \rightarrow Y$.
\end{corollary*}

This follows from the fact that the outputs of finitely many minimal elements in the domain fully characterize a self-stabilizing function. This validates the intuition that the set of all self-stabilizing functions is very limited in comparison to the set of computable functions. We note that our general
self-stabilizing protocol is not efficient, and often specific problems have much faster self-stabilizing protocols. Unsurprisingly the functions we list in Section \ref{ss:examples} admit faster solutions than our protocol. \ifthenelse{\boolean{isFullVersion}}{The proceedings version of this work appears in SSS 2020 \cite{proceedings}.}{}

\subsection{Related Work}

The population protocol model was first introduced by Angluin, Aspnes, Diamadi, Fischer, and Peralta in \cite{Model} to represent a system of mobile finite-state sensors. Dijkstra was the first to formalize the notion of self-stabilization within a distributed system, although the models and problems he discussed imposed different constraints than that of population protocols, such as distinguishing agents with unique identifiers \cite{Dijkstra}. Self-stabilizing population protocols were first formalized in the work of Angluin, Aspnes, Fischer, and Jiang \cite{Angluin}. Their work generated self-stabilizing, constant-space protocols for problems including round-robin token circulation, leader election in rings, and 2-hop coloring in degree-bounded graphs. Moreover their work established a crucial method of impossibility result. Call a class of graphs \textit{simple}, if there does not exist a graph in the class which contains two disjoint subgraphs that are also in the class. Example of this include the class of all rings or the class of all connected degree-d regular graphs. Angluin et. al's work demonstrated that leader election in \textit{non-simple} classes of graphs are impossible. Our paper's impossibility result follows from Angluin et. al's technique, as the class of complete graphs that we work with is non-simple. Other impossibility results come from Cai, Izumi, and Wada \cite{Cai} using \textit{closed sets}, which are sets of states such that a transition on any two of the states results in a state within the set; impossibility in leader election is demonstrated by identifying a closed set excluding the leader state. 

Many self-stabilizing population protocol constructions besides those from \cite{Angluin} tend to give the model additional properties to achieve self-stabilization. Beauquier, Burman, Clement, and Kutten introduce intercommunication speeds amongst agents, captured by the \textit{cover time}; they also add a distinguished, non-mobile agent with unlimited resources called a \textit{base station} \cite{Beauquier2009}. Under this model, Beauquier, Burman, and Kutten design an \textit{automatic tranformer} that takes a population protocol algorithm solving some static problem and transforms it into a self-stabilizing algorithm \cite{Beauquier2009_2}. Izumi, Kinpara, Izumi, and Wada also use this model to create efficient protocols for performing a self-stabilizing count of the number of agents in the network, where there is a known upper bound $P$ on the number of sensors. Their protocol converges under global fairness with $3 \cdot \lceil \frac{P}{2} \rceil$ agent states \cite{Izumi2014}. Fischer and Jiang introduce an \textit{eventual leader detector} oracle into the model that allows self-stabilizing leader election in complete graphs and rings. The complete graph protocol works with local and global fairness conditions, while the ring protocol requires global fairness \cite{Fischer2006}. Beauquier, Blanchard, and Burman extend this work by presenting self-stabilizing leader election in arbitrary graphs when a composition of eventual leader detectors is introduced into the model \cite{Beauquier2013}. Knowledge of the number of agents allows Burman, Doty, Nowak, Severson, and Xu to develop several efficient self-stabilizing protocols for leader election; with no silence or space constraints they achieve optimal expected parallel time of $\mathcal{O}(\log n)$ \cite{burman2020efficient}. Loosely-stabilizing protocols relax self-stabilization to allow more tractable solutions, such as leader election protocols with polylogarithmic convergence time by Sudo, Ooshita, Kakugawa, Masuzawa, Datta, and Larmore \cite{sudo_et_al:LIPIcs:2018:10090}. Self-stabilizing Leader election (under additional symmetry-breaking assumptions) is possible without unique identifiers, as discussed in \cite{Awerbuch,Mayer}, but the communication happens over a fixed graph - unlike population protocols where interaction of agents is arbitrary and interaction pattern is controlled by an adversarial scheduler.
Population self-stabilizing protocols are also related to biological systems self-stabilization, see \cite{G18} for further discussion.

In this work, we do not extend the basic model with any extra abilities. We demonstrate a universal self-stabilizing population protocol for any function $f: \mathcal{X} \rightarrow Y$ where for any $A, B \in \mathcal{X}$, $A \subseteq B$ implies that $f(A) = f(B)$. We do this by using a result from partial order theory by Dickson \cite{Dickson} that states that any set of finite dimensional vectors of natural numbers have finitely many minimal elements under the pointwise partial order. 


\subsection{The Number of Self-Stabilizing Functions Depends on the Number of Minimal Elements}

In Definition \ref{d:rootset} we define the \textit{root set}. Formally, let $\mathcal{X}$ be a (possibly infinite) set of finite multisets over a finite alphabet (e.g. $\mathcal{X} = \{\{a\}, \{a, a\}, \hdots\}$ over alphabet $\Sigma = \{a\}$). The root set is some subset $\mathcal{R} \subseteq \mathcal{X}$ with the following property: for any element of the domain $A \in \mathcal{X}$, there is some element $R \in \mathcal{R}$ such that $R \subseteq A$. We call $R$ a root. Section \ref{SS:RootSet} uses Dickson's Lemma to prove that there exists a unique, finite, and minimally sized root set $\mathcal{R}$. We note the following corollary to the main theorem: to determine the output of any $A \in \mathcal{X}$ for a self-stabilizing function $f: \mathcal{X} \rightarrow Y$, it suffices to identify some root $R \subseteq A$ since $f(R) = f(A)$. In fact the entire output of $f$ is determined by $f(R)$ for each $R \in \mathcal{R}$. The number of possible outputs is then upper bounded by the size of the smallest root set, which is an interesting fact in of its own right.

\begin{corollary*}
    Let $\mathcal{X}$ be a set of finite multisets over a finite alphabet, let $Y$ be a finite output alphabet, and let $\mathcal{R}$ be the minimally sized root set of $\mathcal{X}$. Then $|\mbox{im}(f)| \leq |\mathcal{R}|$.
\end{corollary*}

Fix some domain $\mathcal{X}$ and finite output alphabet $Y$, and let $\mathcal{R}$ be the minimally sized root set of $\mathcal{X}$. The number of self-stabilizing functions $f: \mathcal{X} \rightarrow Y$ are precisely the number of ways to assign an output to each root in $\mathcal{R}$, of which there are at most $N = |Y|^{|\mathcal{R}|}$ ways. In fact, the number of functions is exactly $N$ if and only if every $A \in \mathcal{X}$ has a unique root $R \subseteq A$. This is true because if two roots $R$ and $R'$ are both subsets of $A$, then they must be assigned the same output $f(R) = f(R') = f(A)$; otherwise if every $A$ has a unique root, then there is no overlap and each root has $|Y|$ choices for its output. \ifthenelse{\boolean{isFullVersion}}{See Appendix \ref{a:numfun} for brief discussion on the gap between the number of self-stabilizing functions and computable functions.}{}

\begin{corollary*}
    Let $\mathcal{X}$ be a set of finite multisets over a finite alphabet, let $Y$ be a finite output alphabet, and let $\mathcal{R}$ be the minimally sized root set of $\mathcal{X}$. The number of self-stabilizing functions $f: \mathcal{X} \rightarrow Y$ is finite. Specifically there are at most $|Y|^{|\mathcal{R}|}$ self-stabilizing functions. There are exactly this number of self-stabilizing functions if and only if every $A \in \mathcal{X}$ has a unique root $R \in \mathcal{R}$.
\end{corollary*}

\subsection{Nontrivial Examples of Self-Stabilizing Functions} \label{ss:examples}
If we can restrict the domain $\mathcal{X}$ to exclude some inputs, then it can become easier to generate problems that admit self-stabilizing solutions.
{
\begin{itemize}
    \item In \textbf{Chemical Reaction Networks (CRN)}, it can be desirable to compute boolean circuits. In CRNs one prevalent technique to compute some boolean function $g: \{0,1\}^n \rightarrow \{0,1\}^m$ is to have $n$ different species, each with 2 sub-species each \cite{Cook2009}. That is, we have molecules $s_1^0, s_1^1, s_2^0, s_2^1, \hdots,$ $s_n^0, s_n^1$, where molecule $s_i^j$ signifies that the $i^{th}$ bit has value $j \in \{0,1\}$. All species will appear in the input, but only one sub-species per species will appear (i.e. $s_i^0$ and $s_i^1$ will not both appear, but one of them will). This way of formulating the input allows for self-stabilizing computation of boolean functions! This is because if $A \subseteq B$ for some input multisets, $A$ and $B$, of molecules, they will include the same sub-species and hence have the same output $g(A) = g(B)$. 
    
    \item Generalizing the former example, suppose we have an input alphabet of $n$ symbols, but only $k < n$ of these symbols will ever show up in the population (though a given symbol could occur multiple times). Then we can compute \textit{any} self-stabilizing function of the $k$ present members. 
    
    For instance, suppose we have $k$ different classes of finite-state mobile agents, $a_1, \hdots, a_k$. There are a nonzero number of agents in each class $a_i$, and agents within the same class are running a (possibly not self-stabilizing) population protocol. Eventually every agent in class $a_i$ will converge and be outputting the same common value $o_i$. These outputs $o_1, \hdots, o_k$ will then be the input to our self-stabilizing population protocol. The agents might perform some sort of protocol composition where a tuple of states $(q,s)$ are used for each agent; $q$ would correspond to the agent state in the first protocol, and then $s$ would correspond to the agent state in the self-stabilizing protocol. If the first protocol was also self-stabilizing (it doesn't have to be), then the entire protocol composition would be self-stabilizing as well.
    
    As an example, we could have each class of agents $a_i$ compute some boolean circuit $C_i$ in a self-stabilizing way. Eventually each agent of class $a_i$ will be outputting some $m$ character string over the set $\{0_i, 1_i\}$. Notice that even both agents in class $a_1$ and in class $a_2$ intend to output 0110, the former agents will output $0_11_11_10_1$ and the latter agents will output $0_21_21_20_2$. Once all of these boolean circuit computations are done, all the agents across classes will calculate the majority output in a self-stabilizing way (and they would be able to do so since the outputs $o_1, \hdots, o_k$ are all distinguishable).
    \item Any computable function in which the number of agents is a fixed constant $k$ admits a self-stabilizing solution (there can't be any subsets, so the condition is vacuously true). For instance, distribute bits amongst exactly $k$ agents; we can create a self-stabilizing protocol to output 1 if any permutation of those $k$ bits represents a prime number in binary.
\end{itemize}
} 

\subsection{The Basic Model}
There are different formalizations of the basic population protocol model. We adopt the basic one first introduced by Angluin et. al \cite{Model}, except where we impose that any two agents are allowed to interact.

A \textbf{population protocol} is a tuple $\mathcal{P} = (Q, \Sigma, Y, I, O, \delta)$ where $Q$ is the finite set of agent states; $\Sigma$ is a finite set of input symbols; $Y$ is a finite set of output symbols; $I: \Sigma \rightarrow Q$ is an input function; $O: Q \rightarrow Y$ is the output function; and $\delta: (Q \times \Sigma) \times (Q \times \Sigma) \rightarrow (Q \times Q)$ is the transition function.

Note that population protocols are independent over the size of the system; rather, first a population protocol is specified and then it is run on some set of agents $V$. At the beginning of execution, an input assignment $\alpha: V \rightarrow \Sigma$ is provided, providing each agent an input symbol (we will observe that in complete graphs, we can view input assignments as merely a multiset of inputs). Since our model focuses on the computation of functions, we will enforce that the input does not change (i.e. the input is hardwired into every agent). If an agent is assigned input symbol $\sigma \in \Sigma$, it will determine its starting state via the input function as $I(\sigma)$ (note that in the basic model the input determines the starting states, but in self-stabilization we do not consider starting states). At each time step, a scheduler selects (subject to a global fairness condition) an agent pair $(u,v)$ for interaction; semantically the scheduler is selecting agents $u$ and $v$ to interact, where $u$ is called the \textit{initiator} and $v$ is called the \textit{responder}. Agents $u$ and $v$ will then state transition via $\delta$; letting $q_u, q_v \in Q$ and $\sigma_u, \sigma_v \in \Sigma$ be the states and inputs for $u$ and $v$ respectively, the new respective states will be the output of $\delta((q_u, \sigma_u),(q_v, \sigma_v))$. 

We use the notion of a \textit{configuration} to describe the collective agent states.
\begin{definition}{\textit{Configuration.}}
    Let $V$ be the set of agents and $Q$ be a set of states. A \textit{configuration} of a system is a function $C:V \rightarrow Q$ mapping every agent to its current state.
\end{definition}

When running the protocol we will go through a sequence of configurations. If $C$ and $C'$ are configurations under some population protocol such that $C'$ can follow from a single agent interaction in $C$, we write $C \rightarrow C'$. If a series of interactions takes us from $C$ to $C'$, we write $C \xrightarrow{*} C'$.

\begin{definition}{\textit{Execution.}}
    An \textit{execution} of a population protocol is a sequence of configurations $\mathcal{C} = C_1C_2\hdots$ where for all $i$, $C_i \rightarrow C_{i+1}$.
\end{definition}

The scheduler is subject to a global fairness condition, which states that if a configuration can follow from an infinitely occurring configuration, then it must also occur infinitely often.

\begin{definition}{\textit{Global Fairness Condition} \cite{Model}\textit{.}}
    Let $C$ and $C'$ be configurations such that $C \rightarrow C'$. If $C$ appears infinitely often during an execution, then $C'$ appears infinitely often during that execution.
\end{definition}

At each time step every agent outputs some symbol from $Y$ via output function $O$. If all agents output the same symbol and continue to do so for each time step afterwards, we say the protocol's output is that symbol (interactions may continue, but the output of the agents remain that symbol). Note that when computing functions, the protocol should output the same symbol when the same inputs are provided, irrespective of the globally fair scheduler's behavior. When the population has determined the final output, we say it has \textit{converged}.

\begin{definition}{\textit{Convergence} \cite{Model}\textit{.}}
    A population is said to have \textit{converged} to an output $y \in Y$ during a population protocol's execution if the current configuration $C$ is such that each agent's output is $y$ and $C \rightarrow C'$ implies that $C'$ has each agent with the same output $y$.
\end{definition}

When agents $u$ and $v$ interact, their state transition is a function of both agent's current state and respective inputs; since this is all the information they receive, an agent does not learn the identity of the agent it interacts with, but merely the state and input of that agent. In this sense, agents with the same input that are in the same state are \textit{indistinguishable} from one another. Furthermore, population protocols are independent of the number of agents, making it impossible to design the state set to give every agent a unique identifier, so agents are truly \textit{anonymous}.

As noted earlier, our model accepts input via an input assignment $\alpha: V \rightarrow \Sigma$, where $V$ is the set of agents and $\Sigma$ is our finite input alphabet. It is useful to note, though, that we can actually view our input instead as some finite multiset $A$ over alphabet $\Sigma$, with $|A| = |V|$ \cite{Aspnes}. Though we omit the proof, the idea follows from the fact that any two agents can interact, so which agent gets what input symbol is less important than what input symbols are provided to the system in the first place. We use the notation $m_A(\sigma)$ to denote the \textit{multiplicity}, the number of occurrences, of element $\sigma$ in multiset $A$.

\begin{definition}{\textit{Population Protocol Functions.}}
    Let $f: \mathcal{X} \rightarrow Y$ be a function where $\mathcal{X}$ is a set of multisets over the finite alphabet $\Sigma$. A population protocol $\mathcal{P}$ computes $f$ if and only if for any $A \in \mathcal{X},$ all executions of $\mathcal{P}$ with input $A$ converge to $f(A)$.
\end{definition}

In this paper when we say $\mathcal{P}$ is a population protocol (or simply protocol), we mean that it computes some function $f$. When we don't care about whether a population protocol computes some function, we will refer to it as a \textit{sub-protocol}. This will be useful jargon in our protocol composition in Section \ref{ss:protocolcomp}.


We say that a population protocol computing a function is \textit{self-stabilizing} when it can begin in any configuration and eventually converge (to the same output). Such a function is called a \textit{self-stabilizing} function.

\begin{definition}{\textit{Self-Stabilizing Protocol.}}
    Let $\mathcal{P} = (Q, \Sigma, Y, I, O, \delta)$ be a population protocol computing some function $f:\mathcal{X} \rightarrow Y$. $\mathcal{P}$ is called \textit{self-stabilizing} if and only if for any input multiset $A \in \mathcal{X}$, any set of agents $V$ of cardinality $|A|$, and any starting configuration $C: V \rightarrow Q$, we have that any execution of $\mathcal{P}$ converges to $f(A)$.
\end{definition}



\section{Constructing a Self-Stabilizing Protocol for the Basic Model} \label{s:construct}
We aim to show that not only does Theorem \ref{Self-Stabilizing Population Protocol Theorem} specify \textit{necessary} conditions for computing self-stabilizing functions (shown in \ifthenelse{\boolean{isFullVersion}}{Appendix \ref{a:impossibility}}{the full version of this paper \cite{full}}), but they are also \textit{sufficient} for self-stabilization. We do this by generating a self-stabilizing protocol for computing all functions of the form $f: \mathcal{X} \rightarrow Y$ where for all $A,B \in \mathcal{X}, A \subseteq B \implies f(A) = f(B)$.
To do this, we must introduce a new notion known as the root set of a set of multisets.

\subsection{The Root Set} \label{SS:RootSet}
The inputs for our agents are represented by a multiset of inputs, an element of domain $\mathcal{X}$. We are interested in a kind of subset $\mathcal{R} \subseteq \mathcal{X}$ such that all multisets in $\mathcal{X}$ are a superset of some multiset in $\mathcal{R}$. This section aims to show that all sets of multisets $\mathcal{X}$ actually have a \textit{finite} $\mathcal{R}$.
\begin{definition}{\textit{Root Set and its Roots.}} \label{d:rootset}
    Let $\mathcal{X}$ be a set of finite multisets. A subset $\mathcal{R} \subseteq \mathcal{X}$ is called a \textit{root set} of $\mathcal{X}$ if and only if for all $A \in \mathcal{X}$, there exists $R \in \mathcal{R}$ such that $R \subseteq A$. We call a multiset $R \in \mathcal{R}$ a \textit{root} of $A$.
\end{definition}
\ifthenelse{\boolean{isFullVersion}}{See Appendix \ref{A:RootSetExample} for an example.}{} Notice that $\mathcal{X}$ is always trivially its own root set. However, $\mathcal{X}$ can be infinitely large; for example the set of all nonempty finite multisets on alphabet $\Sigma = \{a\}$ is $\mathcal{X} = \{\{a\}, \{a, a\}, \hdots\}$. However, we are primarily interested in the existence of a \textit{finite} root set over our function $f$'s domain. Dickson's Lemma \cite{Dickson} provides us what we need, though an alternative proof is in \ifthenelse{\boolean{isFullVersion}}{Appendix \ref{A:RootSetProof}}{our full paper \cite{full}}.

First, consider a finite multiset over an alphabet, $\Sigma$, of $n$ elements. An equivalent representation of a multiset is as a vector of multiplicities in $\mathbb{N}^n$. For instance, the multiset $\{a, a, b\}$ on ordered alphabet $\Sigma = \{a, b, c\}$ would be represented by $(2, 1, 0)$. Hence a set of finite multisets on an alphabet of size $n$ could be considered a subset $S \subseteq \mathbb{N}^n$. Consider two vectors $\vec{n}, \vec{m} \in \mathbb{N}^n$, and denote $n_i$ and $m_i$ as the $i^{th}$ component in the corresponding vectors. Define the pointwise partial order $\vec{n} \leq \vec{m} \iff n_i \leq m_i \text{ for all } i$. A minimal element of a subset $S \subseteq \mathbb{N}^n$ is an element that has no smaller element with respect to this partial order. Now we can state Dickson's Lemma.

\begin{lemma}{\textit{Dickson's Lemma}.}
    In every subset $S \neq \emptyset$ of $\mathbb {N} ^{n}$, there is at least one but no more than a finite number of elements that are minimal elements of $S$ for the pointwise partial order.
\end{lemma}

This is equivalent to the existence of a finite root set.

\begin{corollary}
    Let $\mathcal{X}$ be a set of finite multisets over a finite alphabet. $\mathcal{X}$ has a finite root set. In other words, there exists a finite subset $\mathcal{R} \subseteq \mathcal{X}$ such that for any $A \in \mathcal{X}$, there exists $R \in \mathcal{R}$ such that $R \subseteq A$.
\end{corollary}

A interesting corollary is that the root set of minimal size is unique, making it legitimate to speak of \textit{the} minimally-sized root set. See \ifthenelse{\boolean{isFullVersion}}{Appendix \ref{a:uniquerootset}}{the full version of this paper for a proof \cite{full}}.
\begin{corollary}
    Let $\mathcal{X}$ be a set of finite multisets over a finite alphabet. The minimal length root set $\mathcal{R}$ of $\mathcal{X}$ is unique.
\end{corollary}

From order theory, we note that the minimal root set is a strong downwards antichain. We call a subset $A$ of a poset $P$ a strong downwards antichain if $A$ is an antichain (a subset of $P$ with incomparable elements) and no two distinct elements of $A$ have a smaller element in $P$. The proof is in \ifthenelse{\boolean{isFullVersion}}{Appendix \ref{a:strongdownwardsantichain}}{the full version \cite{full}}.
\begin{corollary}
    Let $\mathcal{X}$ be a set of finite multisets over a finite alphabet. The minimal length root set $\mathcal{R}$ of $\mathcal{X}$ is a strong downwards antichain with respect to the subset partial ordering.
\end{corollary}

Now suppose we have a function $f: \mathcal{X} \rightarrow Y$ over multisets such that for $A,B \in \mathcal{X}, A \subseteq B \implies f(A) = f(B)$. We have that there exists a finite, minimal length root set $\mathcal{R} \subseteq \mathcal{X}$. If the input to the population is $A$, it suffices to identify the root $R \in \mathcal{R}$ such that $R \subseteq A$, since the output would be $f(R) = f(A)$.

\subsection{Self-Stabilizing Population Protocol Construction} \label{SS:Protocol}
Before we formally specify a universal self-stabilizing population protocol, it is much more helpful to first understand how it works at a high level. As a reminder, we will be working with functions of the form $f: \mathcal{X} \rightarrow Y$, where $\mathcal{X}$ is some set of finite multisets on a finite alphabet and $Y$ is the finite output alphabet. The function $f$ is also assumed to satisfy the following property: $$\forall A,B \in \mathcal{X}, A \subseteq B \implies f(A) = f(B).$$

We need to determine how to design our protocol to compute $f$ in a self-stabilizing manner. Since $\mathcal{X}$ is a set of finite multisets, it has a finite root set; let $\mathcal{R} \subseteq \mathcal{X}$ be the minimally sized, finite root set. $\mathcal{R}$ will be given some arbitrary fixed ordering so that we can index into it. Loosely speaking, we can compute $f$ by having each agent iterate through every root in the finite root set to see if a given root is a subset of the population's input $A$. 

As an agent iterates through the root set, how can it tell if the current root is indeed a subset of the population's input multiset, $A$? Suppose an agent has guessed that root $R_i \subseteq A$. Consider some other root $R_j$ where $f(R_i) \neq f(R_j)$. Naturally there are some symbols that occur more often in $R_j$ than in $R_i$, and vice versa. Consider a symbol $\sigma$ that occurs $n$ times in $R_j$ and $m$ times in $R_i$, where $n > m$. Suppose agents start counting how many other agents they see with input $\sigma$, and manage to identify there are at least $n$ of them. Now further suppose that this is the case for all such $\sigma$; that is, every $\sigma$ occurring more often in $R_j$ has at least that many inputs distributed in the population. Then we now know that $f(R_i) \neq f(A)$ by way of contradiction. For any symbol $\sigma$ occurring $n$ times in $R_j$ and $m$ times in $R_i$, we have two cases. If $n > m$, then we know that input multiset $A$ has at least $n$ instances of $\sigma$. If $n \leq m$, since $R_i \subseteq A$ by our hypothesis then we also know that there are at least $n$ instances of $\sigma$. Therefore we simultaneously have that $R_j \subseteq A$ and $R_i \subseteq A$, which is a contradiction since $f(A) = f(R_j) \neq f(R_i)$! Therefore our initial assumption was wrong and $R_i$ is not a subset of the input multiset $A$, and so we should increment our index $i$ to guess root $R_{i+1}$. Note that we could also change our index $i$ to become $j$, and the authors believe this would be a faster protocol; however since this is a paper about computability, we leave this optimization as an observation.

On the other hand if $R_i$ really is a subset of $A$, then the existence of such a $R_j$ would be impossible. If our counters are all initialized to 0, then no agent counter would ever increment high enough to identify such a $R_j$. However since the protocol may start in an arbitrary configuration, we have to make sure to reset counters whenever an agent increments and guesses a new root.

To formalize this idea, we need convenient notation. We define $MORE_{i,j}$ to be the set of symbols occurring more often in $R_j$ than in $R_i$ when $f(R_i) \neq f(R_j)$.
$$
    MORE_{i,j} = 
    \begin{cases}
    \{ \sigma \mid m_{R_i}(\sigma) < m_{R_j}(\sigma) \} & f(R_i) \neq f(R_j) \\
    \emptyset & \text{otherwise} 
    \end{cases}.
$$
Each agent will have a table indexed by $i$ and $j$, where each entry is a binary string of length $|MORE_{i,j}|$. The $k^{th}$ bit of this string is set to 1 when the agent counts that the $k^{th}$ symbol of $MORE_{i,j}$ occurs in the population as often as it does in $R_j$; otherwise it is 0. Once any entry in this table becomes a (nonempty) bit string of all 1's, then the agent tries a new root. To keep track of these counts, each agent will also maintain a nonnegative integer $count$; when two agents with the same $count$ and same input symbol $\sigma$ meet, the responder will increment its value. To keep the states finite, the count is bounded above by the maximum multiplicity that occurs in the root set. It's important to note that the fact that the root set is finite is crucial to keeping the number of states here finite.

\subsection{Sub-Protocols for Protocol Composition} \label{ss:protocolcomp}
The self-stabilizing population protocol we construct will be a protocol composition $A \times B \times C$, where the input is given to $A$, the input to $B$ is the output of $A$, the input to $C$ is the output of $B$, and the composition output is the output of $C$. With the previous discussion as our design motivation, we decompose our protocol into three distinct sub-protocols:
{ 
\begin{enumerate}
    \item $SymbolCount$. This sub-protocol implements a simple modular counting mechanism, where agents with the same input symbol $\sigma$ compare their counts. If two agents with the same count meet, then the responder increments its count (modulo the maximum multiplicity of any symbol in the root set, $M$, to keep things finite). If there are at least $k < M$ agents with the same symbol, then some agent must eventually have their count \textit{at least} $k - 1$ by the Pigeonhole Principle, irrespective of initial configuration. The converse is only true if all the counts are initialized to 0, which is problematic if we are designing a self-stabilizing protocol that can initialize in any configuration. We will circumvent this by having the larger protocol composition reset this counter whenever moving on to the next root.
    \item $WrongOutput?$. This sub-protocol maintains a table indexed by $i,j \in \{0,1,$ $\hdots, |\mathcal{R}| - 1\}$, where $\mathcal{R}$ is the minimally sized root set. Each entry will be a binary string of length $|MORE_{i,j}|$, where $MORE_{i,j}$ is the set of all input symbols occurring more often in $R_j$ than in $R_i$. If the $k^{th}$ symbol of $MORE_{i,j}$ occurs in the population at least as often as it does in $R_j$, then the corresponding bit in the binary string is set to 1. Notice that this table is only useful when it is initialized with all entries as binary strings of all 0's. Again, we will circumvent this by having the larger protocol composition reset the table whenever moving on to the next root.
    \item $RootOutput$. This sub-protocol uses an index $root \in \{0, 1, \hdots, |\mathcal{R}| - 1\}$. Adopting the notation from the previous bullet and letting $root = i$, this protocol increments $root$ if there is a $j$ such that the $(i,j)$ entry in the table has a binary string of all 1's. The protocol composition will use the incrementing of $root$ to signal that the other sub-protocol states should reset.
\end{enumerate}
} 
We now list the three sub-protocols below.

\begin{definition}{$SymbolCount$.}
    Let $f: \mathcal{X} \rightarrow Y$ be a function over finite multisets on a finite alphabet $\Sigma$, and let $\mathcal{R}$ be a finite and minimally sized root set of $\mathcal{X}$. Each agent has a state called \emph{count}, where \emph{count} $\in \{0, 1, \hdots, M\}$ and $M$ is the maximum multiplicity of any symbol in the root set. Let $M = \max_{R \in \mathcal{R}, \sigma \in \Sigma}m_{R}(\sigma).$ Each agent takes as input some $\sigma$ from some input multiset $A \in \mathcal{X}.$ When an agent meets another agent with the same $\sigma$ and same \emph{count}, one of them will increment their \emph{count} modulo $M$. This guarantees that if there are $n$ agents with symbol $\sigma$, then eventually one agent will have \emph{count} $\geq n - 1$. 
                            $$(\emph{count}, \sigma) , (\emph{count}, \sigma) \rightarrow (\emph{count}, \emph{count} + 1 \mod M)$$
    An agent in state \emph{count} with input $\sigma$ outputs $(\emph{count}, \sigma)$. 
\end{definition}
Note that even though output is a function of just the state, we can formally allow $SymbolCount$ agents to output their input symbol as well, since we could take our current states $Q$ and define a new set of states via the cross product $Q' = Q \times \Sigma$. Then transitions on this could be defined in a similar way, where we would also have to make agents transition into a state that reflects their own input symbol (which can happen after every agent interacts once).

For $WrongOutput?$, our transitions need to satisfy two properties. 
\begin{enumerate}
    \item First, it is natural to have agents share their tables with each other to share their collected information about the counts of the inputs. Whenever an agent meets another agent, they bitwise OR their tables. 
    \item Fix some ordering on $MORE_{i,j}$ and denote the $k^{th}$ symbol by $\sigma_k$. Consider the bitstring entry at index $i$ and $j$. By our previous discussion we want the $k^{th}$ bit to be 1 if the number of occurrences of $\sigma_k$ in the population is at least its multiplicity in $R_j \in \mathcal{R}$.
    
    This can be enforced by setting this bit to 1 when the $count$ of some agent with input $\sigma_k$ is sufficiently high enough; then the first property will ensure this bit is set for the other agents by bitwise OR'ing tables. This can formally be accomplished by bitwise OR'ing the agent's table with an indicator table of the same dimension. Let $i$ and $j$ be the indices into the table, $k$ be the index into the bitstring entry, $\sigma_k$ be the $k^{th}$ symbol of $MORE_{i,j}$, and $count$ and $\sigma$ be the agent's inputs. Define
    $$    
        INDICATOR_{i,j,k} = 
        \begin{cases}
            1 & \text{ if } \sigma = \sigma_k \text{ and } count \geq m_{R_j}(\sigma) - 1 \\
            0 & \text{ otherwise }
        \end{cases}.
    $$
    We will denote this table as $INDICATOR(count, \sigma)$ to be explicit about the agent's inputs $count$ and $\sigma$. Note that the $count$ is 0-indexed, which is why we subtract by 1. Also, since agents can start in arbitrary states, every agent has to transition once to have the table updated with $INDICATOR$.
    
\end{enumerate}
\begin{definition}{$WrongOutput?$}
    Let $f: \mathcal{X} \rightarrow Y$ be a function over finite multisets on a finite alphabet, and let $\mathcal{R}$ be a finite and minimally sized root set of $\mathcal{X}$. Each agent has a state called $\varA{HAS-MORE}$, a table indexed by $i$ and $j$ where each entry is a binary string of length $|MORE_{i,j}|$ (as described in previous discussion).
                                    $$\varA{HAS-MORE}_{i,j} \in \{0,1\}^{|MORE_{i,j}|}.$$
    Each agent takes as input (from $SymbolCount$) some $\emph{count} \in \{0, 1, \hdots, M\}$, where $M$ is the maximum multiplicity of any symbol in the root set, and some $\sigma$ from the input multiset $A \in \mathcal{X}$. 
    
    Denote bitwise OR with symbol $\lor$. Our transition rule is
    $$\displaylines{
        (\varA{HAS-MORE}^1, (\emph{count}, \sigma)) , (\varA{HAS-MORE}^2, (\emph{count}', \sigma')) \cr \downarrow \cr (\varA{HAS-MORE}^3 \lor INDICATOR(\emph{count}, \sigma) , \cr\varA{HAS-MORE}^3 \lor INDICATOR(\emph{count}', \sigma')) \cr
        \text{ where } \varA{HAS-MORE}^3 = \varA{HAS-MORE}^1 \lor \varA{HAS-MORE}^2.
    }
    $$
    An agent outputs their $\varA{HAS-MORE}$ table.
\end{definition}

$RootOutput$ will have an integer $root$ that is an index into the root set. An agent with input $\varA{HAS-MORE}$ increments its $root$ modulo $|\mathcal{R}|$ if there is a $j$ such that $\varA{HAS-MORE}_{root, j}$ is all 1's. Of course, this would mean that the state $root$ keeps cycling every time an agent with such an input interacts. When we define the overall protocol composition after, this will be resolved by resetting the previous two sub-protocols' states when $root$ increments. Notice that agents don't care about the states of the other agents in an interaction; instead the behavior depends on how sub-protocol $WrongOutput?$ changes its output over time. Notice that we don't allow protocols to have changing inputs, but our 3 protocol composition allows two of the sub-protocols to have a changing input as the states of the other sub-protocols change.
\begin{definition}{$RootOutput$.}
    Let $f: \mathcal{X} \rightarrow Y$ be a function over finite multisets on a finite alphabet, and let $\mathcal{R}$ be a finite and minimally sized root set of $\mathcal{X}$. Each agent has a state called $\emph{root}$, where $\emph{root} \in \{0, 1, \hdots, |\mathcal{R}| - 1\}.$ Each agent takes as input $\varA{HAS-MORE}$, a table indexed by $i$ and $j$ where each entry is a binary string of length $|MORE_{i,j}|$ (as described in previous discussion).
                            $$\varA{HAS-MORE}_{i,j} \in \{0,1\}^{|MORE_{i,j}|}.$$
    Our transition rules are:
    $$\displaylines{
        ((\emph{root}_1, \varA{HAS-MORE}^1), (\emph{root}_2, \varA{HAS-MORE}^2)) \rightarrow (\emph{root}'_1, \emph{root}'_2) \cr
        \text{where } i' = 
        \begin{cases}
            i + 1 \mod |\mathcal{R}| & \text{if } \exists j \text{ s.t. } \varA{HAS-MORE}_{i, j} = 1^{|MORE_{i,j}|} \cr
            i & \text{otherwise}
        \end{cases}.
        }
    $$
    An agent outputs $f(R_{\emph{root}})$.
\end{definition}

Putting these three sub-protocols together, we get the \textit{SS-Protocol}, a self-stabilizing population protocol for $f$.

\begin{definition}{\textit{SS-Protocol}.}
    Let $f: \mathcal{X} \rightarrow Y$ be a function over finite multisets on a finite alphabet, and let $\mathcal{R}$ be a finite and minimally sized root set of $\mathcal{X}$. Define \textit{SS-Protocol} as protocol composition $$SymbolCount \times WrongOutput? \times RootOutput,$$ where we define protocol composition in the beginning of Section \ref{ss:protocolcomp}. We additionally modify this composition's transition function so that whenever an agent's $RootOutput$ state $\emph{root}$ gets incremented, then
    \begin{itemize}
        \item The $\emph{count}$ state from $SymbolCount$ becomes 0.
        \item The $\varA{HAS-MORE}$ state from $WrongOutput?$ becomes a table of all 0's.
    \end{itemize}
    When this modified transition occurs, we say the agent has been \textit{reset} and is in a \textit{reset state} (note there are multiple reset states as we allow $\emph{root}$ to be arbitrary).
\end{definition}

Say the input to the population is $A$ with root $R_i$. If the protocol ever incorrectly outputs $f(R_j) \neq f(R_i)$, it will recognize this because the protocol would count enough symbols in $MORE_{i,j}$ to make $\varA{HAS-MORE_{i,j}}$ a bitstring of all 1's. If the protocol outputs $f(R_i)$ and begins in a reset state, then there will never be a bitstring entry with all 1's. This protocol composition is a self-stabilizing protocol for $f$. See \ifthenelse{\boolean{isFullVersion}}{Appendix \ref{a:correct} for the full proof}{the full version of this paper for the proof \cite{full}}.

\begin{theorem*}[Self-Stabilizing Population Protocol Theorem]\label{T:SSPP}
    Let $f: \mathcal{X} \rightarrow Y$ be a function computable with population protocols on a complete interaction graph, where $\mathcal{X}$ is a set of multisets. Then
            $f \text{ has a self-stabilizing protocol} \iff (\forall A,B \in \mathcal{X}, A \subseteq B \implies f(A) = f(B)).$
\end{theorem*}

\section{Conclusions and Generalizations}
A function $f: \mathcal{X} \rightarrow Y$ in the basic computational model of population protocols is self-stabilizing if and only if for any multisets $A$ and $B$ in the domain, $A \subseteq B \implies f(A) = f(B)$. The principle insight yielding the forward implication is that we cannot have the scheduler isolate a subpopulation and have that subpopulation restabilize to another output. The converse holds because we only need to parse the root set of the domain and find a root that is a subset of the population to determine what the output is, as our protocol does.

The notion of a root set should be applicable to arbitrary interaction graphs as well. Angluin et. al \cite{Angluin} demonstrated that leader election in \textit{non-simple} classes of graphs are impossible, which is effectively applying the idea that different subgraphs may converge on different answers. We can view input assignments as interaction graphs where the nodes are the input symbols; if there exists a subgraph of an input assignment that maps to a different output under $f$, then $f$ could not admit a self-stabilizing protocol. We believe that the converse should hold as well. If we consider the class of all graphs of input assignments, we can use the subgraph relation $\subseteq$ as our partial order and find its corresponding minimal elements. Unfortunately there are infinite minimal elements in the class of rings, leading to an infinite root set. Perhaps taking a quotient on the domain of possible input assignments (e.g. calling all rings equivalent) may lead to a finite root set, though this is a subject for future research, as well as on characterizing self-stabilization in general population protocols and other distributed models.

\section*{Acknowledgments}
 We thank the anonymous reviewers for their helpful comments. This work is supported in part by DARPA under Cooperative Agreement No: HR0011-20-2-0025, NSF-BSF Grant 1619348, US-Israel BSF grant 2012366, \linebreak Google Faculty Award, JP Morgan Faculty Award, IBM Faculty Research Award, Xerox Faculty Research Award, OKAWA Foundation Research Award, B. John Garrick Foundation Award, Teradata Research Award, and Lockheed-Martin Corporation Research Award. The views and conclusions contained herein are those of the authors and should not be interpreted as necessarily representing the official policies, either expressed or implied, of DARPA, the Department of Defense, or the U.S. Government. The U.S. Government is authorized to reproduce and distribute reprints for governmental purposes not withstanding any copyright annotation therein.


\bibliographystyle{splncs04}
\bibliography{mylib}

\ifthenelse{\boolean{isFullVersion}}{

\newpage
\appendix
\section{Number of Self-Stabilizing Functions vs Computable Functions - Brief Discussion} \label{a:numfun}
Recall our corollary that the for a fixed domain $\mathcal{X}$ and fixed output alphabet $Y$, the number of self-stabilizing functions $f: \mathcal{X} \rightarrow Y$ is finite.

\begin{corollary*}
    Let $\mathcal{X}$ be a set of finite multisets over a finite alphabet, let $Y$ be a finite output alphabet, and let $\mathcal{R}$ be the minimally sized root set of $\mathcal{X}$. The number of self-stabilizing functions $f: \mathcal{X} \rightarrow Y$ is finite. Specifically there are at most $|\mathcal{R}|^{|Y|}$ self-stabilizing functions. There are exactly this number of self-stabilizing functions if and only if every $A \in \mathcal{X}$ has a unique root $R \in \mathcal{R}$.
\end{corollary*}

On the other hand, when $\mathcal{X}$ is all finite multisets on a finite alphabet, we know that the set of computable predicates are exactly the set of semilinear predicates, of which there are (countably) infinitely many. If $f: \mathcal{X} \rightarrow \{0,1\}$ is a computable predicate, the predicate $g(A) = f(A)$ on some restricted domain $\mathcal{X}' \subseteq \mathcal{X}$ will also be a computable predicate (by just running the same protocol for $f$ on the restricted domain). There are cases where some computable predicates $f$ and $h$ are not equivalent on $\mathcal{X}$ but are equivalent on $\mathcal{X}'$ (e.g. $f$ outputs 1 if there are only $a$ symbols, and $h$ always outputs 1). However if $\mathcal{X}'$ only removes finitely many inputs from $\mathcal{X}$, then we can guarantee that $f$ and $h$ continue to not be equivalent on $\mathcal{X}'$ (this is because the set of accepted/rejected inputs can both be written as finite unions of linear sets, and one can show that the set of inputs that $f$ and $h$ disagree on is infinite). Thus we have that for (countably) infinite number of subsets $\mathcal{X}' \subseteq \mathcal{X}$, in which only a finite number of inputs are removed, there are a (countably) infinite number of computable predicates. This argument is \textit{not} a proof that there is a gap between the number of self-stabilizing predicates and number of computable predicates, but merely helps corroborate the intuition that there is one.

\section{Self-Stabilization Impossibility from Angluin} \label{a:impossibility}
Angluin et. al's impossibility result for self-stabilization informally asserts that if a self-stabilizing protocol converges in a population, any subgraph of the population should converge to the same answer had it run under the same protocol. If not, the scheduler could isolate that subgraph after the convergence of the population, causing agents in the subgraph to self-stabilize and change their output, thereby contradicting convergence.

This proof technique is known, but we include it below for sake of completeness.
\begin{theorem}
    \label{T:Subgraph}
    Let $\mathcal{P} = (Q, \Sigma, Y, I, O, \delta)$ be a self-stabilizing population protocol computing some function $f: \mathcal{X} \rightarrow Y$, where $\mathcal{X}$ is a set of multisets over finite alphabet $\Sigma$. Then $$\forall A,B \in \mathcal{X}, A \subseteq B \implies f(A) = f(B)$$
\end{theorem}
\begin{proof}
    Suppose $f(A) \neq f(B)$ for some $A, B \in \mathcal{X}$ where $A \subseteq B$. We run protocol $\mathcal{P}$ on an arbitrary configuration of $|B|$ agents where each agent is given an input symbol from $B$. Since $\mathcal{P}$ is self-stabilizing and computes $f$, it will eventually converge to the output $f(B)$.
    
    Since the scheduler can act arbitrarily for a finite amount of time (by global fairness), after convergence the scheduler can isolate members\footnote{A globally fair scheduler can isolate any subset of the population. If $C$ is an infinitely occurring configuration, let $C'$ be the configuration that follows in which finitely many interactions occur only within that subpopulation. Then $C \rightarrow C'$ and so $C'$ occurs infinitely often.} of the population whose input symbols together form multiset $A$. Since this protocol is self-stabilizing and $A \in \mathcal{X}$, the protocol will eventually have this population self-stabilize and converge to $f(A)$. However this is a contradiction, since by the definition of convergence, all agents must continue to output $f(B)$ for all time after converging and yet this subpopulation is now outputting $f(A) \neq f(B)$. Thus by contradiction, it must be that $f(A) = f(B).$ 
\end{proof}

The principle notion is that we must disallow the possibility of a subpopulation converging to a different answer than what the entire population converged to. In a setting where only certain ordered pairs of agents can interact, which interactions are allowed could be described by an interaction graph on the agents. In this setting we still cannot have a subgraph converge in isolation to a different answer. See Conclusion and Generalizations for more discussion.

\section{Proofs of Correctness} \label{a:correct}
We now prove that \textit{SS-Protocol} is a self-stabilizing population protocol for $f: \mathcal{X} \rightarrow Y$ where for any $A,B \in \mathcal{X}$, $A \subseteq B \implies f(A) = f(B)$. To begin, we want to argue that if there are $n$ agents with the same input, then the counting sub-protocol will achieve a $count$ of at least $n - 1$ (where we subtract 1 due to 0-indexing). This will be useful because we need $count$ to be sufficiently high enough so that we can identify when to increment $root$. This lemma is almost immediate, but there is some subtle nuance since the counting mechanism may reset before becoming sufficiently large enough.

\begin{lemma} \label{l:countbound}
    Let $f: \mathcal{X} \rightarrow Y$ be a function over finite multisets on a finite alphabet, and let $\mathcal{R} = \{R_0, R_1, \hdots, R_{|\mathcal{R}| - 1} \}$ be a finite and minimally sized root set of $\mathcal{X}$. Let $A \in \mathcal{X}$ be the multiset whose elements are dispersed amongst the agents. Let $n \leq M$, where $M$ is the maximum multiplicity of any symbol in the root set. Let $C$ be a configuration in \textit{SS-Protocol} where there are at least $n$ agents with input $\sigma$. There is some configuration $C'$, with $C \xrightarrow{*} C'$, where one of these agents has $SymbolCount$ state $\emph{count} \geq n - 1$.
\end{lemma}
\begin{proof}
    Consider $n$ of these agents with input $\sigma$ in some configuration $C$, and suppose all of these agents have $0 \leq count < n - 1$. By the Pigeonhole Principle there will be at least two agents with the same $count$, and so it is valid to keep scheduling interactions between agents with the same $count$ while no agent has $count \geq n - 1$. If none of these agents reset when interacting with each other, then it follows that interactions between agents with duplicate values of $count$ will eventually raise some agent's $count$ to be at least $n - 1$ (the subtraction by 1 due to 0-indexing). If an agent does reset when only interacting with agents with the same input, it can only do so finitely many times. Suppose some agent with $root = i$ resets during these interactions; it must be because the input $\sigma$ has a corresponding bit in some $\varA{HAS-MORE}_{i,j}$ entry, and hence $\sigma \in MORE_{i,j}$. Either the agent stops resetting, or it keeps resetting until $root$ increments to $root = j$; then since $MORE_{i,j}$ is the complement of $MORE_{j,i}$, $\sigma$ won't have a corresponding bit in $\varA{HAS-MORE}_{j,i}$. Thus continued interactions with agents with input $\sigma$ could not lead to another reset. Consequently all agents with input $\sigma$ must eventually stop resetting, and so the agents with duplicate values of $count$ will eventually drive some agent's $count$ to be at least $n - 1$. This final configuration $C'$, where some agent with input $\sigma$ has $count \geq n - 1$, follows from $C$.
\end{proof}

The next lemma allows us to argue that any agent that is outputting the wrong answer (due to an incorrect choice of $root$) will eventually identify this, ultimately leading to $root$ being incremented. The proof idea is that the input multiset $A$ must have a root $R_j \subseteq A$, and the agents with symbols from $R_j$ will have $count$ sufficiently high enough to make the agent increment $root$ and reset. A corollary to this is that eventually such an agent will start outputting the right answer as it keeps iterating through the root set. It's also crucial to note that every time $root$ increments, the agent enters a reset state. 

\begin{lemma} \label{l:iterate}
    Let $f: \mathcal{X} \rightarrow Y$ be a function over finite multisets on a finite alphabet, and let $\mathcal{R} = \{R_0, R_1, \hdots, R_{|\mathcal{R}| - 1} \}$ be a finite and minimally sized root set of $\mathcal{X}$. Let $A \in \mathcal{X}$ be the multiset whose elements are dispersed amongst the agents. Consider some execution of \textit{SS-Protocol}. Suppose there is an infinitely occurring configuration in which an agent has $RootCount$ state $\emph{root} = i$ such that $f(R_i) \neq f(A)$. Then there is an infinitely occurring configuration where this agent has $\emph{root} = i + 1 \mod |\mathcal{R}|$ and is in a reset state.
\end{lemma}
\begin{proof}
    For sake of notational convenience, we will let $root = i$ throughout this proof. Suppose there is some infinitely occurring configuration with an agent with $RootOutput$ state $i$ such that $f(R_i) \neq f(A)$; that is, the agent is outputting the wrong answer. As $f(R_i)$ is not the right output, it must be that $R_i \not \subseteq A$. By definition of the root set there exists some different root $R_j$ such that $R_j \subseteq A$, which in turn implies $f(R_j) = f(A).$ Therefore $f(R_i) \neq f(R_j)$, so $MORE_{i,j}$ is the set of symbols occurring more often in $R_j$ than in $R_i$. Note that $MORE_{i,j}$ must be nonempty since otherwise $R_j \subseteq R_i$, which contradicts the minimality of the root set since the root set would be smaller if we excluded $R_i$. For each $\sigma \in MORE_{i,j}$, there will be at least as many agents with input $\sigma$ as there are occurrences in $R_j$, since $R_j \subseteq A$. In particular there will be at least $m_{R_j}(\sigma)$ agents with input $\sigma$. By Lemma \ref{l:countbound} and global fairness, there is an infinitely occurring configuration where one of these agents has $count \geq m_{R_j}(\sigma) - 1$. This agent will hence have the component corresponding to $\sigma$ in $\varA{HAS-MORE}_{i,j}$ set to 1. There will be an agent like this for each $\sigma \in MORE_{i,j}$; after our original agent interacts with each of them and bitwise OR's her $\varA{HAS-MORE}$ table with them, it would set her $\varA{HAS-MORE}_{i,j}$ bitstring entry to $1^{|MORE_{i,j}|}$. The original agent would increment her $RootCount$ state to $root = i + 1 \mod |\mathcal{R}|$, which would reset the agent. This resultant configuration follows from an infinite configuration, and so by global fairness it also must be infinitely occurring.
\end{proof}

Once an agent with input $\sigma$ resets, we want their $count$ to reflect a lower bound of the number of agents with input $\sigma$. Specifically, suppose all agents have reset at least once. Then if an agent with input $\sigma$ has $count = n$, we want it to be the case that there are at least $n + 1$ agents with input $\sigma$ in the population. Though this fact is intuitive, it turns out to be a little cumbersome to prove. We leave the details to Appendix \ref{a:countlowerbound}.
\begin{lemma} \label{l:countlowerbound}
    Let $f: \mathcal{X} \rightarrow Y$ be a function over finite multisets on a finite alphabet. Let $A \in \mathcal{X}$ be the multiset whose elements are dispersed amongst the agents. Suppose that there is an infinitely occurring configuration where all agents have reset at least once. If an agent with input $\sigma$ has $count = n$, then there must be at least $n + 1$ agents with input $\sigma$ in the population.
\end{lemma}

One of the challenges of demonstrating protocol correctness is that the population can begin in any arbitrary configuration of states. For instance, it may be the case that $root$ corresponds to the correct root, but the $count$ and $\varA{HAS-MORE}$ table are so poorly initialized that we end up erroneously incrementing $root$! The following lemma captures what happens when we have a favorable initialization: when all agents have reset at least once, then no $count$ for input $\sigma$ can ever overestimate the actual number of agents with input $\sigma$ in the population. When we know that all agents have reset at least once, we can guarantee convergence.

\begin{lemma} \label{l:converge}
    Let $f: \mathcal{X} \rightarrow Y$ be a function over finite multisets on a finite alphabet, and let $\mathcal{R} = \{R_0, R_1, \hdots, R_{|\mathcal{R}| - 1} \}$ be a finite and minimally sized root set of $\mathcal{X}$. Let $A \in \mathcal{X}$ be the multiset whose elements are dispersed amongst the agents. Suppose that there is an infinitely occurring configuration where all agents have reset at least once. Then the protocol will converge with output $f(A)$.
\end{lemma}
\begin{proof}
    Suppose we are in an infinitely occurring configuration $C$ where all agents have reset at least once. This means that all agents have had their counters reset to $count = 0$ at some point. By Lemma \ref{l:countlowerbound} this means that an agent with $count = n$ and input $\sigma$ implies that there must be at least $n + 1$ agents with input $\sigma$ in the population. 
    
    Fix an arbitrary agent, where we aim to show that this agent will output $f(A)$ forever. To do this, we will find it useful to:
    \begin{itemize}
        \item Reset the agent (again) so that its $WrongOutput?$ table $\varA{HAS-MORE}$ has all 0's as its entries. This gets rid of any new bits that might've been set since the agent's last reset.
        \item Make the agent's $RootCount$ state $root = i$, where $R_i \subseteq A$.
    \end{itemize} 
    If the agent currently has the wrong output, then repeated use of Lemma \ref{l:iterate} will increment the $root$ until the agent has the correct output, $f(R_{root}) = f(A)$. Now we can assume the agent has the correct output in some infinitely occurring configuration. From this configuration, it either eventually outputs correctly for all time (in which case we are done) or it keeps changing output by repeatedly incrementing $root$ and resetting. Therefore we can now assume that there is an infinitely occurring configuration where the agent has $root = i$ with table $\varA{HAS-MORE}$ having all 0 bitstrings. Since $root = i$, the agent correctly outputs $f(R_i) = f(A)$.
    
    To show that this agent will eventually converge on this answer, consider the case where the agent changes its output again. For this to happen, it must be the case that its $RootCount$ state $root = i$ was incremented yet again. This means that for some $j$, the $WrongOutput?$ state entry $\varA{HAS-MORE}_{i,j}$ becomes all 1's. Note that the definition states that $MORE_{i,j}$, the set of symbols occurring more often in $R_j$ than in $R_i$, must be \textbf{nonempty} for this to happen. Since the agent was just reset with 0's in every component of $\varA{HAS-MORE}$, this can only happen if for each $\sigma \in MORE_{i,j}$ there is some agent with a $count$ sufficiently high enough: specifically, $count \geq m_{R_j}(\sigma) - 1$. By Lemma \ref{l:countlowerbound}, it must be that there are at least $m_{R_j}(\sigma)$ agents with input symbol $\sigma$ in the population. 
    
    We claim that now we have that $R_j \subseteq A$, which we show by considering two cases. If $\sigma \in MORE_{i,j}$, then we already know there are at least $m_{R_j}(\sigma)$ agents with input symbol $\sigma$, so the count of $\sigma$ in $R_j$ is at most the count in $A$. If $\sigma \not \in MORE_{i,j}$, then the number of occurrences of $\sigma$ in $R_j$ is at most the number in $R_i$, which is at most the number in $A$ since $R_i \subseteq A$. Thus the multiplicity of each element in $R_j$ is at most the respective multiplicity in $A$. Therefore $R_j \subseteq A$ and hence $f(R_i) = f(A) = f(R_j)$. This is a contradiction because this would make $MORE_{i,j}$ the empty set.
\end{proof}

Finally, we can extend convergence to unfavorable initializations. Essentially, Lemma \ref{l:iterate} tells us that agents that output the wrong answer will eventually reset. If agents outputting the right answer eventually reset as well, we would be done due to Lemma \ref{l:converge}. We accomplish this by way of contradiction, where the protocol not converging would lead to at least one agent iterating through the root set. If this iterating agent is always able to set a bitstring entry to all 1's for any choice of $root$, then any other agent that hasn't reset yet should be able to as well. 
\begin{theorem*}[Self-Stabilizing Population Protocol Theorem]
    Let $f: \mathcal{X} \rightarrow Y$ be a function computable with population protocols on a complete interaction graph, where $\mathcal{X}$ is a set of multisets. Then
    \begin{align*}
        f \text{ has a self-stabilizing protocol} \iff (\forall A,B \in \mathcal{X}, A \subseteq B \implies f(A) = f(B)).
    \end{align*}
\end{theorem*}
\begin{proof}
    The forward direction is shown in Appendix \ref{a:impossibility}. We claim that \textit{SS-Protocol} is a self-stabilizing protocol for $f$. Consider some infinitely occurring configuration $C$ in an execution of \textit{SS-Protocol} with input multiset $A \in \mathcal{X}$; we will argue that all agents will eventually output $f(A)$ forever.
    
    By Lemma \ref{l:iterate}, all agents that currently output the wrong answer will eventually be reset. If all agents that currently output the correct answer eventually reset, then Lemma \ref{l:converge} proves the protocol converges. Otherwise, there is an agent with $root = j$ currently outputting the correct answer that will never reset; call this agent $NoReset$. If the protocol converges anyway, we are done. If not, then there is some other agent that eventually outputs the wrong answer an infinite number of times; call this agent $RootIncrementor$. By Lemma \ref{l:iterate} $RootIncrementor$ will have to output correctly eventually, and so this agent oscillates between correct and wrong outputs. This can only happen if the $RootIncrementor$'s $root$ keeps getting incremented for all time, resetting the agent each time. This means after cycling through values, $RootIncrementor$ will eventually also have its $root = j$, just like $NoReset$. Since $RootIncrementor$ will increment and reset yet again after interacting with some agents, it must be the case that $NoReset$ could also interact with these same agents to satisfy the conditions for an increment and reset (both $RootIncrementor$ and $NoReset$ will have an entry in their table $\varA{HAS-MORE}$ with all 1's). Therefore by global fairness $NoReset$ does eventually reset, a contradiction. Thus the protocol must converge.
    
\end{proof}

\section{Reset Agents and $count$ Lower Bound} \label{a:countlowerbound}
Suppose all agents during some execution of our protocol have reset (i.e. at some point had $count = 0$). Then if an agent with input $\sigma$ has $count = n$, then there must be at least $n + 1$ agents with input $\sigma$ in the population. Though this feels intuitive, the proof takes some care.

Suppose an agent $a$ with input $\sigma$ starts off with $count = 0$. It might interact with another agent $b$ with input $\sigma$ and $count = 0$, which will lead to $a$ incrementing to $count = 1$. Subsequently it may interact with an agent $c$ with input $\sigma$ and $count = 1$, leading $a$ to increment to $count = 2$. However it could be that $b = c$. So we need to actually not double count $b$, but count the latent agent $d$ that allowed $b$ to increment its $count$. We do this by maintaining a set of agents $A$ that will add agents via specific rules as the $count$ of $a$ increments. For instance:
\begin{enumerate}
    \item In the beginning, $A = \{a\}$.
    \item When $a$ increments to 1 after interacting with $b$, $A = \{a, b\}$.
    \item When $b$ increments to 1 after interacting with $d$, $A = \{a, d\}$.
    \item When $a$ increments to 2 after interacting with $b$, $A = \{a, b, d\}$.
\end{enumerate}

However, another wrinkle is that agents might reset their $count$ at any point due to the other subprotocols at play. Moreover, we need to guarantee that every time $a$ interacts with another agent, it can't be any agent from $A$. For this last fact, we must argue that it is always the case that $a$ has the strictly largest $count = k$ in $A$; doing this requires characterizing the distribution of $count$ amongst the agents of $A$. In particular there is always no other agent with $count = k$, at most 1 other agent with $count \geq k - 1$, at most 2 other agents with $count \geq k - 2$, and so on. This characterization of the distribution excludes $a$ from ever interacting with another agent in $A$, so the incrementing of $count$ always corresponds to meeting a new agent.
\begin{lemma*}
    Let $f: \mathcal{X} \rightarrow Y$ be a function over finite multisets on a finite alphabet. Let $A \in \mathcal{X}$ be the multiset whose elements are dispersed amongst the agents. Suppose that there is an infinitely occurring configuration where all agents have reset at least once. If an agent with input $\sigma$ has $count = n$, then there must be at least $n + 1$ agents with input $\sigma$ in the population.
\end{lemma*}
\begin{proof}
    First we show that if a single agent has been reset, then $count = n$ for this agent means $n + 1$ agents have the same input symbol. If all agents are reset once, then this must be the case for every agent, and we are done.
    
    Suppose an agent $a$ with input $\sigma$ has been reset so that $count = 0$. Using a set $A$, we will add new agents with input $\sigma$ as this agent's $count$ increments. So when $a$ resets, set $A \leftarrow \{a\}$. For the rest of this proof, we will refer to $a$ as the \textit{main} agent, any agent in $A - \{a\}$ as a \textit{inside} agent, and all other agents as \textit{outside} agents. As interactions continue to occur, we will maintain the following two rules. First, say an interaction occurs between an inside agent $b$ and an outside agent $c$; if $b$ increments its $count$, then replace $b$ with $c$ in $A$ (i.e. $A \leftarrow A - \{b\} \cup \{c\}$). Notice that this operation keeps the number of agents in $A$ with $b$'s old $count$ the same, since $b$ and $c$ had the same $count$ before the interaction. Second, say an interaction occurs between our main agent $a$ and some outside agent $b \not \in A$. If $a$ increments its $count$, then add $b$ to $A$ (i.e. $A \leftarrow A \cup \{b\}$). 
    
    To do our proof, we first prove the following invariant on the distribution of $count$ on the inside agents: there are 0 inside agents with the same $count$ as $a$, at most 1 inside agent with $count$ at least one less than $a$, at most 2 inside agents with $count$ at least two less than $a$, and so on. Formally, suppose main agent $a$ has $count = k$. Then there are at most $m$ inside agents with $count \geq k - m$, where $0 \leq m \leq k$. We see how this invariant is maintained in all possible interactions that change the $count$ of agents in $A$:
    \begin{itemize}
        \item When main agent $a$ resets, this is trivial since $A = \{a\}$. This is $A$'s initial state.
        \item Suppose an inside agent with $count = j$ resets. The number of inside agents with $count \geq 0$ trivially stays the same. In all other cases, the number of inside agents either decreases by 1 or stays the same.
        \item Suppose the main agent increments its $count$ by meeting another agent. Notice that since there are 0 inside agents with $count \geq k$, this interaction must be with an outside agent. If the main agent increments its $count$ to $k + 1$, the outside agent is added to $A$ with $count = k$. Now we need to show that there are at most $m$ inside agents with $count \geq (k + 1) - m$ for $0 \leq m \leq k + 1$. Before the interaction we had that there were at most $m - 1$ inside agents with $count \geq k - (m - 1) = (k + 1) - m$ for $1 \leq m \leq k + 1$, so making this outside agent an inside agent would now make it at most $m - 1 + 1 = m$ inside agents. The only unconsidered case is $m = 0$, but this is straightforward: if there was an agent with $count \geq k + 1$ after the interaction then it must have been there before the interaction, which is impossible given our invariant.
        \item Suppose an inside agent increments its $count$ by meeting another agent. As noted earlier, the other agent cannot be the main agent due to our invariant. If it is an outside agent, then our rule states that $A$ replaces the inside agent with the outside agent; as noted earlier, this maintains the $count$ distribution of $A$ and hence maintains the invariant. Now suppose an inside agent meets an inside agent with the same $count = k - m$, resulting in one incrementing to $count = k - m + 1$, where $2 \leq m \leq k$. Notice we exclude $m = 0, 1$ since there are never two inside agents both with $count = k$ or $count = k - 1$. Since this only changes the number of inside agents with $count = k - m$ and $count = k - m + 1$, this does not change the number of inside agents with $count$ at least $0, 1, \hdots, k - m - 1, k - m + 2, k - m + 3, \hdots, k$. If the invariant is violated, it can only be because there are now too many agents with $count = k - m + 1$. Specifically a violation means there must be at least $m$ agents with $count = k - m + 1$; this means that before the interaction there were at least $m - 1$ agents with $count = k - m + 1$ and 2 agents with $count = k - m$. This gives a total of at least $m + 1$ agents with $count \geq k - m$ before the interaction, a contradiction. Thus the invariant is not violated.
    \end{itemize} 
    
    Therefore at all points this invariant on the inside agents holds. A corollary is that when $a$'s $count = k$, then $|A| = k + 1$ and hence there are $k + 1$ agents with input $\sigma$. When $count = 0$, then $A$ is the singleton $\{a\}$. Suppose agent $a$ just incremented to $count = k + 1$, having met some other agent $b$ with $count = k$. By the invariant, no inside agent could have that $count$, so $b \not \in A$. Therefore $b$ gets added to $A$, making $|A| = k + 1$.

\end{proof}

\section{The Root Set}
\subsection{Root Set Example} \label{A:RootSetExample}
For instance, take $\mathcal{X}$ over alphabet $\Sigma = \{a, b, c, d, e, f\}$ as
            $$\mathcal{X} = \{ \{a, a, b\}, \{a, a, b, b, c\}, \{e, e, e, f, f, f, b, d \}, \{d\} \}.$$
A root set $\mathcal{R} \subseteq \mathcal{X}$ could be
                            $$\mathcal{R} = \{ \{a, a, b\},  \{d\} \},$$
where $\{a, a, b\}$ and $\{d\}$ are the roots.

\subsection{Self-Contained Proof of Finite Root Set} \label{A:RootSetProof}
The language of the following proof is simplified, without loss of generality, when using vectors. So we will use the following equivalent definition of a vector root set for the next proof. See Appendix \ref{A:VectorRootSetExample} for an example.
\begin{definition}{\textit{Vector Root Set and its Vector Roots.}}
    Let $\mathcal{Y}$ be a set of $n$-dimensional vectors with nonnegative integer components. A subset $\mathcal{R} \subseteq \mathcal{Y}$ is called a vector root set of $\mathcal{Y}$ if and only if for all $\vec{a} = (a_1, \hdots, a_n) \in \mathcal{Y}$, there exists a $\vec{r} = (r_1, \hdots, r_n) \in \mathcal{R}$ such that for all $i$, $r_i \leq a_i$. We call such a vector $\vec{r}$ a vector root of $\vec{a}$.
\end{definition}

The following lemma asserts a finite root set of $\mathcal{X}$ on alphabet $\Sigma = \{\sigma_1, \hdots, \sigma_k\}$ exists if and only if a finite vector root set of
    $$\mathcal{Y} = \{(m_A(\sigma_1), \hdots, m_A(\sigma_k)) \mid A \in \mathcal{X}\}$$
exists. The intuition is that $A \subseteq B$ if and only if $m_A(\sigma) \leq m_B(\sigma)$ for all $\sigma$ in the alphabet. 
\begin{lemma*}
    \label{L:VectorRoot}
    Let $\mathcal{X}$ be a set of finite multisets on an alphabet $\Sigma  = \{\sigma_1, \hdots, \sigma_k\}$. Let $\mathcal{Y}$ be a set of $k$-dimensional vectors generated from $\mathcal{X}$ via
        $$\mathcal{Y} = \{(m_A(\sigma_1), \hdots, m_A(\sigma_k)) \mid A \in \mathcal{X}\}.$$
    Then $\mathcal{X}$ has a finite root set if and only if $\mathcal{Y}$ has a finite vector root set.
\end{lemma*}
\begin{proof}
    Let $\mathcal{R}_\mathcal{X} \subseteq \mathcal{X}$ be a finite root set of $\mathcal{X}$. Define
        $$\mathcal{R}_\mathcal{Y} = \{ (m_R(\sigma_1), \hdots, m_R(\sigma_k)) \mid R \in \mathcal{R}_\mathcal{X} \}.$$
    Certainly $\mathcal{R}_\mathcal{Y} \subseteq \mathcal{Y}$ since $\mathcal{R}_\mathcal{X} \subseteq \mathcal{X}$. For any vector $\vec{a} = (m_A(\sigma_1), \hdots, m_A(\sigma_k)) \in \mathcal{Y}$, by the definition of root set there is a root $R \in \mathcal{R}_\mathcal{X}$ such that $R \subseteq A$. Then $m_R(\sigma_i) \leq m_A(\sigma_i)$ for all $i$, and so $(m_R(\sigma_1), \hdots, m_R(\sigma_k)) \in \mathcal{R}_\mathcal{Y}$ is the vector root for $\vec{a}$. Thus $\mathcal{R}_\mathcal{Y}$ is a finite vector root set.
    
    Conversely (and via a symmetric argument), let $\mathcal{R}_\mathcal{Y}$ be a finite vector root set. Define
        $$\mathcal{R}_\mathcal{X} = \{R \mid (m_R(\sigma_1), \hdots, m_R(\sigma_k)) \in \mathcal{R}_\mathcal{Y} \}.$$
    Certainly $\mathcal{R}_\mathcal{X} \subseteq \mathcal{X}$ since $\mathcal{R}_\mathcal{Y} \subseteq \mathcal{Y}$. For any multiset $A \in \mathcal{X}$, there is a vector root $(m_R(\sigma_1), \hdots, m_R(\sigma_k)) \in \mathcal{R}_\mathcal{Y}$ such that for all $i$, $m_R(\sigma_i) \leq m_A(\sigma_i)$. Therefore $R \in \mathcal{R}_\mathcal{X}$ is a root for $A$, and thus $\mathcal{R}_\mathcal{X}$ is a finite root set.
\end{proof}

With that in mind, showing that a finite vector root set of $\mathcal{Y}$ always exists would imply that a finite root set of $\mathcal{X}$ always exists. The proof uses strong induction on the dimension of the vector, $k$, to construct a finite vector root set. We use the following inductive step. Fix a single vector $\vec{v} = (v_1, \hdots, v_{k+1}) \in \mathcal{Y}$ that shall be included in the vector root set we are constructing (vector $\vec{v}$ will be fixed for the whole argument). If there is a $\vec{u} = (u_1, \hdots, u_{k+1}) \in \mathcal{Y}$ such that $u_i \geq v_i$ for all $i$, then $\vec{v}$ will be a root of $\vec{u}$ in the root set. So we focus our attention only on those $\vec{u} \in \mathcal{Y}$ where there exists an $i$ such that $u_i < v_i$; call the set of these vectors $\mathcal{Y}' \subseteq \mathcal{Y}$. 

To illustrate the next portion of our argument, we present an overview augmented with an example. We aim to partition $\mathcal{Y}'$ into finitely many different classes where all vectors in a single class have a common component; we can ignore this component, generate a finite vector root set by the inductive hypothesis, and then reinsert the common component to generate a finite vector root set for the class. The union of all these finite vector root sets constructs a vector root set for $\mathcal{Y'}$. 

Given parameters $1 \leq i \leq k + 1$ and $0 \leq m < v_i$, we define the set $$F_{i,m} = \{ \vec{u} = (u_1, \hdots, u_{k+1}) \in \mathcal{Y} \mid u_i = m\}$$ that constitutes all vectors in $\mathcal{Y}$ whose $i^{th}$ component is exactly $m$. For example, let $$\mathcal{Y'} = \{ (a,b,c) \mid a,b,c \geq 1 \text{ where exactly one of $a$, $b$, or $c$ is 1}\}.$$ Suppose we have $i = 1$ and $m = 2$ (meaning the first component is always a 2) as follows: $$F_{1,2} = \{ (2,1,2), (2,1,3), (2, 1, 4), \hdots \} \cup \{ (2, 2, 1), (2, 3, 1), (2, 4, 1), \hdots \}.$$ We define a corresponding set $F'_{i,m}$ that is identical to $F_{i,m}$ except it ignores the $i^{th}$ component of every vector in $F_{i,m}$. In this case, $$F'_{1, 2} = \{ (1,2), (1,3), (1, 4), \hdots \} \cup \{ (2, 1), (3, 1), (4, 1), \hdots \}.$$ Since this reduces the dimension of the vectors to $k$, we can get a vector root set $\mathcal{R'}_{i,m}$ by our inductive hypothesis. Here we can take the vector root set as $$\mathcal{R}'_{1,2} = \{ (1, 2), (2, 1) \}.$$ Given that the $i^{th}$ component of the vectors in $F_{i,m}$ are all value $m$, we can add back the $i^{th}$ component to each $\vec{r'} \in \mathcal{R}'_{i,m}$ to get a vector root set, $\mathcal{R}_{i,m}$, for $F_{i,m}$. $$\mathcal{R}_{1,2} = \{ (2, 1, 2), (2, 2, 1) \}.$$

As we are considering vectors $\vec{u} = (u_1, \hdots, u_{k+1}) \in \mathcal{Y}$ where there exists an $i$ such that $u_i < v_i$, $\vec{u}$ will have a root in $\mathcal{R}_{i,u_i}$. The constructed \textit{finite} root set is therefore $$\mathcal{R} = \{\vec{v}\} \cup \bigcup_{1 \leq i \leq k + 1, 0 \leq m < v_i} \mathcal{R}_{i,m}.$$ This proof is formally written in Appendix \ref{A:FiniteVectorRootSet}.

\begin{lemma}[\textit{Finite Root Set Lemma}]
    \label{L:FiniteVectorRootSet}
    Let $\mathcal{Y}$ be a set of $k$-dimensional vectors with nonnegative components. $\mathcal{Y}$ has a finite vector root set. In other words, there exists a finite subset $\mathcal{R} \subseteq \mathcal{X}$ such that for any $\vec{a} = (a_1, \hdots, a_k) \in \mathcal{Y}$, there exists $\vec{r} = (r_1, \hdots, r_k) \in \mathcal{R}$ such that $r_i \leq a_i$ for all $i$.
\end{lemma}
\begin{corollary}
    \label{C:FiniteRootSet}
    Let $\mathcal{X}$ be a set of finite multisets over a finite alphabet. $\mathcal{X}$ has a finite root set. In other words, there exists a finite subset $\mathcal{R} \subseteq \mathcal{X}$ such that for any $A \in \mathcal{X}$, there exists $R \in \mathcal{R}$ such that $R \subseteq A$.
\end{corollary}

\subsection{Vector Root Set Example} \label{A:VectorRootSetExample}
Consider the previous example from Appendix \ref{A:RootSetExample} in the language of vector root sets, where we convert each multiset into a $|\Sigma|$-dimensional vector of multiplicities. Take $\mathcal{Y}$ over \textit{ordered} alphabet $\Sigma = \{a, b, c, d, e, f\}$ as 
$$\mathcal{Y} = \{ (2,1,0,0,0,0), (2,2,1,0,0,0), (0,1,0,1,3,3), (0,0,0,1,0,0) \}.$$ A vector root set $\mathcal{R} \subseteq \mathcal{Y}$ could be $$\mathcal{R} = \{ (2,1,0,0,0,0), (0,0,0,1,0,0) \},$$ where $(2,1,0,0,0,0)$ and $(0,0,0,1,0,0)$ are the roots. Analogously, $\mathcal{Y}$ is trivially its own root set, but it also could have been infinitely large; we are searching for finite vector root sets.

\subsection{Proof of Finite Root Set Lemma} \label{A:FiniteVectorRootSet}
\begin{lemma*}[\textit{Finite Root Set Lemma}]
    Let $\mathcal{Y}$ be a set of $k$-dimensional vectors with nonnegative components. $\mathcal{Y}$ has a finite vector root set. In other words, there exists a finite subset $\mathcal{R} \subseteq \mathcal{X}$ such that for any $\vec{a} = (a_1, \hdots, a_k) \in \mathcal{Y}$, there exists $\vec{r} = (r_1, \hdots, r_k) \in \mathcal{R}$ such that $r_i \leq a_i$ for all $i$.
\end{lemma*}
\begin{proof}
    This proof will be done by strong induction on the dimension of the vectors in $\mathcal{Y}$, which we will denote $k$. For $k = 1$ (i.e. the vectors have only one component), $\mathcal{Y}$ is just a set of nonnegative integers, in which clearly the singleton subset containing the smallest element is the root set. 
    
    Now assume the hypothesis holds for all $k' \leq k$, and let $\mathcal{Y}$ now have vectors of dimension $k + 1$ (where the components are still nonnegative). We will say a vector $\vec{v} = (v_1, \hdots, v_{k+1})$ \textit{dominates} vector $\vec{u} = (u_1, \hdots, u_{k+1})$ if for all $i$, $v_i \geq u_i$. Fix some $\vec{v} = (v_1, \hdots, v_{k+1}) \in \mathcal{Y}$ for the rest of the proof. If we include $\vec{v}$ in our root set, any vector in $\mathcal{Y}$ that dominates $\vec{v}$ will have $\vec{v}$ as its vector root by definition. Therefore we restrict our attention to only those vectors $\vec{u} = (u_1, \hdots, u_{k+1}) \in \mathcal{Y}$ such that for some $i$, $u_i < v_i$.
    
    Given $i$ and $m$ such that $1 \leq i \leq k + 1$ and $0 \leq m < v_i$, define the set $$F_{i,m} = \{ \vec{u} = (u_1, \hdots, u_{k+1}) \in \mathcal{Y} \mid u_i = m\}.$$ We define a corresponding set $F'_{i,m}$ that is identical to $F_{i,m}$ except it ignores the $i^{th}$ component of every vector in $F_{i,m}$; since this reduces the dimension of the vectors to $k$, we can get a finite vector root set $\mathcal{R'}_{i,m}$ by our inductive hypothesis. Given that the $i^{th}$ component of the vectors in $F_{i,m}$ are all value $m$, we can add back the $i^{th}$ component to each $\vec{r'} \in \mathcal{R}'_{i,m}$ to get a vector root set, $\mathcal{R}_{i,m}$, for $F_{i,m}$.
    
    More formally, for any $\vec{u} = (u_1, \hdots, u_{k+1}) \in F_{i,m}$, consider the corresponding vector $\vec{u'} = (u_1, \hdots, u_{i-1}, u_{i+1}, \hdots, u_{k+1}) \in F'_{i,m}$. This has a root $\vec{r'} = (r_1, \hdots, r_{i-1}, r_{i+1}, \hdots, r_{k+1}) \in \mathcal{R'}_{i,m}$, so that $\vec{r} = (r_1, \hdots, r_{i-1}, m, r_{i+1}, \hdots, r_{k+1}) \in \mathcal{R}_{i,m}$. Since $u_j \geq r_j$ for $j \neq i$ (by definition of vector root set $\mathcal{R'}_{i,m}$) and $u_i = m = r_i$, we have that $u_j \geq r_j$ for all $j$. Thus $\mathcal{R}_{i,m}$ is a finite vector root set for $F_{i,m}$.
    
    This procedure can be taken for $1 \leq i \leq k + 1$ and $0 \leq m < v_i$, and we can generate the \textit{finite} set:
        $$\mathcal{R} = \{\vec{v}\} \cup \bigcup_{1 \leq i \leq k + 1, 0 \leq m < v_i} \mathcal{R}_{i,m}.$$
    Thus any vector $\vec{u} \in \mathcal{Y}$ that dominates $\vec{v}$ would have $\vec{v}$ as a vector root; otherwise there is an $i$ such that $u_i < v_i$, and hence it would have a vector root in $\mathcal{R}_{i, u_i}$. Thus, $\mathcal{R}$ is a finite vector root set.
\end{proof}
\begin{corollary*}
    Let $\mathcal{X}$ be a set of finite multisets over a finite alphabet. $\mathcal{X}$ has a finite root set. In other words, there exists a finite subset $\mathcal{R} \subseteq \mathcal{X}$ such that for any $A \in \mathcal{X}$, there exists $R \in \mathcal{R}$ such that $R \subseteq A$.
\end{corollary*}

\subsection{Root Set Corollaries}
\subsubsection{Uniqueness of Minimal Finite Root Set} \label{a:uniquerootset}
\begin{corollary} 
    Let $\mathcal{X}$ be a set of finite multisets over a finite alphabet. The minimal length root set $\mathcal{R}$ of $\mathcal{X}$ is unique.
\end{corollary}
\begin{proof}
    Suppose there exists two root sets of minimal size, $\mathcal{R}$ and $\mathcal{R}'$. Then there is some root $R_i \in \mathcal{R}$ such that $R_i \not \in \mathcal{R}'$. Since $R_i \in \mathcal{X}$, it must have a root $R' \in \mathcal{R}'$ such that $R' \subseteq R_i$. Since $R' \in \mathcal{X}$, it must have a root $R_j \in \mathcal{R}$ such that $R_j \subseteq R'$. Therefore we have that $R_j \subseteq R' \subseteq R_i$, for some $R_i, R_j \in \mathcal{R}$. If $R_i = R_j$ then we have that $R_i \subseteq R' \subseteq R_i$, which is a contradiction since this means $R' = R_i \in \mathcal{R}'$. If $R_i \neq R_j$ then $\mathcal{R}$ is not of minimal size since $\mathcal{R} - \{R_i\}$ is also a root set of $\mathcal{X}$. Thus any root in $\mathcal{R}$ is also in $\mathcal{R}'$; since $|\mathcal{R}| = |\mathcal{R}'|$ by assumption, the sets must be equal.
\end{proof}

\subsubsection{The Minimal Root Set is a Strong Downwards Antichain} \label{a:strongdownwardsantichain}
\begin{corollary} 
    Let $\mathcal{X}$ be a set of finite multisets over a finite alphabet. The minimal length root set $\mathcal{R}$ of $\mathcal{X}$ is a strong downwards antichain with respect to the subset partial ordering.
\end{corollary}
\begin{proof}
     If there were $R_i, R_j \in \mathcal{R}$ such that $R_i \subseteq R_j$, then $\mathcal{R} - \{R_j\}$ would be a smaller root set, contradicting its minimal length. Thus due to its minimality, $\mathcal{R}$ is an antichain. Suppose there were two distinct roots $R_i, R_j \in \mathcal{R}$ such that some element $A \in \mathcal{X}$ is a subset of both $R_i$ and $R_j$. Then $(\mathcal{R} \cup \{A\}) - \{R_i, R_j\}$ is a smaller root set, contradicting its minimal length. Thus it is a strong downwards antichain.
\end{proof}
 
}{}
 
\end{document}